\documentclass[journal]{IEEEtran}
\usepackage{times,epsfig,epstopdf,latexsym,multirow,array,amssymb,graphicx,multicol,stfloats,textcomp,psfrag}
\usepackage[cmex10]{amsmath}
\usepackage{eqparbox,slashbox}
\usepackage{amsfonts}
\usepackage{here}
\usepackage{rawfonts}
\usepackage[latin1]{inputenc}
\usepackage[T1]{fontenc}
\usepackage{calc}
\usepackage{url}
\usepackage{enumerate}
\usepackage{color}
\usepackage{upref}
\usepackage{epic,eepic}
\usepackage{dsfont}
\usepackage{comment}
\usepackage{cite}

\newtheorem{theorem}{{Theorem}}

\newtheorem{corollary}{{ Corollary}}

\usepackage{float}
\usepackage{stfloats}
\newcounter{tempequationcounter}
\begin{document}

\title{\LARGE{Jamming Energy Allocation in Training-Based Multiple Access Systems} }

\author{
\authorblockN{Hamed Pezeshki, \emph{Student Member, IEEE,} 
Xiangyun Zhou,  \emph{Member, IEEE}, 
and Behrouz Maham, \emph{Member, IEEE}
} 
\\
    \thanks{ 
    The work of X. Zhou was supported by the Australian Research Council's Discovery Projects funding scheme (Project No. DP110102548).     
    
    Hamed Pezeshki and Behrouz Maham are with the School of Electrical and
    Computer Engineering, University of Tehran, Tehran 14395-515, Iran (e-mails: \{h.pezeshki, bmaham\}@ut.ac.ir).
    
    X. Zhou is with the Research School of Engineering, The Australian National
    University, ACT 0200, Australia (e-mail: {xiangyun.zhou@anu.edu.au).}}
\vspace{-9mm}
}

{}
\maketitle

\IEEEpeerreviewmaketitle
\begin{abstract}
We consider the problem of jamming attack in a multiple access channel with training-based transmission. First, we derive upper and lower bounds on the maximum achievable ergodic sum-rate which explicitly shows the impact of jamming during both the training phase and the data transmission phase. Then, from the jammer's design perspective, we analytically find the optimal jamming energy allocation between the two phases that minimizes the derived bounds on the ergodic sum-rate. Numerical results demonstrate that the obtained optimal jamming design reduces the ergodic sum-rate of the legitimate users considerably in comparison to fixed power jamming.
\end{abstract}
\vspace{-0.4cm}
\section{Introduction}
\IEEEPARstart{T}{he} problem of jamming exists in many practical wireless communication scenarios and can culminate in dissatisfactory Quality of Service (QoS) in commercial networks. For communication systems requiring training signals to facilitate channel estimation at the receiver, smart jamming attack during the training phase can result in a detrimental effect \cite{acm}, \cite{han}. Recently, optimal jamming energy allocation against training-based, single-user point-to-point systems was studied in \cite{Sean}, where the jammer optimally allocates its energy to attack both the training and data transmission phases. In this letter, we focus on a multiuser system over a multiple access channel (MAC) \cite{main}, \cite{Soysal} and design the jamming energy allocation from the attacker's perspective.

In the multiuser MAC setting, the jammer's design depends on the system parameters of \textit{all} the users. For instance, the optimal power level for jamming a particular user's training signal depends not only on the training and data transmission powers of this user, but also on the power levels of all the other users in both training and data transmission. Hence, the jammer's design in the multiuser system is much more challenging compared with the single-user system because of the complex interdependence between a large number of parameters of all users.

To facilitate the smart jammer design, we first derive tractable upper and lower bounds on the ergodic sum-rate of the MAC system, taking into account the different jamming powers used during training and data transmission. Then, we find analytical solutions to the jamming energy allocation which minimizes the bounds on the ergodic sum-rate. Our numerical results compare the ergodic sum-rate achieved under the optimal jamming energy allocation with that under fixed power jamming, i.e., the non-optimized case. A significant rate reduction from the non-optimized case to the optimized case is observed, ranging from 35\% to 90\%, which clearly shows the potential impact of smart jamming attack in training-based MAC systems.
\vspace{-0.3cm}
\section{System Model}
We consider a block-fading single-antenna narrowband multiple access channel (MAC) with $K$ users. This structure forms the legitimate users of our model. We also assume that there exists a smart jammer which aims to reduce the ergodic sum-rate of the legitimate users by injecting an artificial noise. The input-output relationship of the legitimate users is given by\footnote[1]{Note that the fading gain of the channel from the jammer to the receiver is absorbed into the jamming signal $w$ \cite{Sean}, \cite{corr}. In order for the received jamming signal to be Gaussian, the jammer needs to carefully design its signal by taking the channel statistics into account.}
\begin{equation}\label{1}
y=\sum_{k=1}^{K}h_kx_k+n+w
\end{equation}
where $x_k$ and $y$ are the transmitted symbol by user $k$, $k \in \mathcal{K}=\{1,...,K\}$, and the received symbol, respectively. $h_k\sim \mathcal{CN}(0,1)$ is the channel coefficient of user $k$, $n\sim \mathcal{CN}(0,1)$ is the additive white Gaussian noise and finally $w$ is the artificial noise injected by the jammer, which is again assumed to be Gaussian, since for Gaussian channels the Gaussian distribution is the best from the jammer's design point of view \cite{corr}. We have normalized the variances of $h_k$ and $n$ to one, without loss of generality.

We consider a training-based blockwise transmission, where each block in every user starts with a training phase followed by a data transmission phase. The training signals of the users are assumed to be non-overlapping in time \cite{main}, \cite{Soysal}\footnote[2]{The assumption of non-overlapping training sequences is often used in multiuser MAC systems \cite{main}, \cite{Soysal}, in order to avoid interference between the pilot symbols of the users. With the non-overlapping training sequences, the base station can employ simple estimation methods to obtain accurate channel estimates, which result in a good capacity or error probability performance. The training overhead is relatively small in slow fading channels with a small number of MAC users, while it increases as the number of users increases.}.
Each user sends $T_{t_k}$ pilot symbols; therefore, the total training duration is $T_t=\sum_{k=1}^{K}T_{t_k}$. During the training phase, the receiver performs a linear minimum mean square error (LMMSE) estimation \cite{Poor} to estimate the channels of the users using the $T_t$ pilot symbols. During the remaining $T_d$ symbols, data transmission occurs and the receiver uses the estimated channel gains to detect the data.
Therefore, the total block length for every user is $T=T_t+T_d$ symbol periods. We consider a block-fading scenario where the channel remains constant during one block and changes to an independent and identically distributed realization in the next block. Furthermore, each user transmits its pilot and data symbols with power $P_{t_k}$ and $P_{d_k}$, respectively. On the attacker side, the jammer also uses different jamming power, i.e., $P_{wt_k}$ and $P_{wd}$, to jam the training symbols of user $k$ and the data symbols of all the users, respectively. 

We assume that the jammer can acquire the values of $P_{t_k}$ and $P_{d_k}$ for all the users by eavesdropping on the feedback channel from the receiver, e.g., base station, to the users. This technique has been proposed in the literature to obtain information about the attacked communication system. For instance, it is assumed in \cite{Shafiee} that the jammer can infer the user channel states by eavesdropping on the feedback channel from the receiver to the users in fading multiple access channels. In the timing perspective, we also assume that the jammer is capable of eavesdropping on the feedback or command signalling between base station and the users. The jammer is assumed to be able to listen to the signalling exchanges during the synchronization stage of the legitimate network and hence obtain the timing information, i.e., the values of $T_{t_k}$ and $T_d$.

We consider the LMMSE estimator. For the channel of every user, we have $h_k=\hat{h}_k+\tilde{h}_k$, where $\hat{h}_k$ and $\tilde{h}_k$ are the estimated channel and channel estimation error of user $k$, respectively, and we have \cite{Poor}
\begin{equation}\label{2}
 \sigma_{\hat{h}_k}^2=\frac{\frac{P_{t_k}}{1+P_{wt_k}}T_{t_k}}{1+\frac{P_{t_k}}{1+P_{wt_k}}T_{t_k}}\hspace{2mm}\textnormal{and}\hspace{2mm} \sigma_{\tilde{h}_k}^2=\frac{1}{1+\frac{P_{t_k}}{1+P_{wt_k}}T_{t_k}}
\end{equation}
where $\sigma_{\hat{h}_k}^2$ and $\sigma_{\tilde{h}_k}^2$ are the variances of $\hat{h}_k$ and $\tilde{h}_k$, respectively.
\vspace{-0.6cm}
\section{Problem Formulation}
We consider a smart jamming design in an energy-constrained scenario and denote the average power budget of the jammer by $P_w$. Hence, we have $P_wT=\sum_{k=1}^{K}P_{wt_k}T_{t_k}+P_{w_d}T_d$. 
We assume that the jammer allocates the ratio $\zeta_{t_k}$ and $\zeta_d$ of its total energy to jam the training signal of user $k$ and the data symbols of all the users, respectively. Hence, we have
\begin{equation}\label{3}
 \zeta_{t_k}=\frac{P_{wt_k}T_{t_k}}{P_wT}\hspace{2mm}\textnormal{and}\hspace{2mm}
 \zeta_d=\frac{P_{w_d}T_d}{P_wT}.
\end{equation}
Define $\zeta_t=\sum_{k=1}^{K}\zeta_{t_k}$, which is the ratio of the energy that the jammer allocates to the training phase. Then, we have $\zeta_t+\zeta_d=1.$
In this letter, we look from the jammer's point of view and find the optimal values of $\{\zeta_{t_k}\}_{k\in \mathcal{K}}$ and $\zeta_d$ which minimize the achievable ergodic sum-rate, when the legitimate users' system parameters are known, i.e., for \textit{any} given values of $\{P_{t_k}\}_{k\in \mathcal{K}}$, $\{P_{d_k}\}_{k\in \mathcal{K}}$ and $\{T_{t_k}\}_{k\in \mathcal{K}}$.

When the CSI at the receiver is noisy, the optimum input distribution and hence the exact expression for the ergodic sum-capacity are unknown. However, the maximum achievable sum-rate with Gaussian signalling is given by  \cite{Soysal}
\begin{equation}\label{4}
R=\frac{T-T_t}{T}\mathbb{E}_{\hat{{h}}_k}\left[\log_2\left(1+\frac{\sum_{k=1}^{K}\frac{P_{d_k}}{1+P_{wd}}|\hat{h}_k|^2}{1+\sum_{k=1}^{K}\sigma_{\tilde{h}_k}^2\frac{P_{d_k}}{1+P_{wd}}}\right)\right]
\end{equation}
where $\frac{T-T_t}{T}$ reflects the amount of time spent during the training phase and leads to the capacity loss. Now, we further obtain an upper bound and a lower bound on $R$.

\textit{1) Upper-Bound:} Considering the fact that $\log_2(1+x)$ is a concave function of $x$, we use Jensen's inequality to derive an upper bound on $R$. Hence, we have
\begin{equation}\label{5}
\begin{split}
R &\leq \frac{T-T_t}{T}\log_2\left(1+\frac{\mathbb{E}_{\hat{{h}}_k}\left\{\sum_{k=1}^{K}\frac{P_{d_k}}{1+P_{wd}}|\hat{h}_k|^2\right\}}{1+\sum_{k=1}^{K}\sigma_{\tilde{h}_k}^2\frac{P_{d_k}}{1+P_{wd}}}\right)\\
&=\frac{T-T_t}{T}\log_2\left(1+\rho(\{\zeta_{t_k}\}_{k\in \mathcal{K}},\zeta_d)\right)=R_{UB}
\end{split}
\end{equation}
where
\begin{equation}\label{6}
\rho(\{\zeta_{t_k}\}_{k\in \mathcal{K}},\zeta_d)=\frac{\sum_{k=1}^{K}\frac{P_{d_k}}{1+P_{wd}}\sigma_{\hat{h}_k}^2}{1+\sum_{k=1}^{K}\sigma_{\tilde{h}_k}^2\frac{P_{d_k}}{1+P_{wd}}}.
\end{equation}
\textit{2) Lower-Bound:} Define the vector $[x_1,...,x_k]$ of multiple variables. Then, $\log_2(1+\sum_{k=1}^{K}a_ke^{x_k})$ is a convex function on $\mathbb{R}^K$ for arbitrary $a_k>0$ \cite[Lemma 3]{convex}. Hence by applying Jensen's inequality in \eqref{4}, we have
 \begin{figure*}[!b]
 \vspace{-0.4cm}
 \normalsize
 \hrulefill
 \setcounter{tempequationcounter}{\value{equation}}
 \setcounter{equation}{12}
 \begin{equation}\label{f1}
\zeta_{t_k}^*=\frac{1}{P_wT}\left(\frac{\sqrt{P_wT\gamma(\zeta_d^*)P_{d_k}P_{t_k}T_{t_k}^2\left(1+\gamma(\zeta_d^*)\sum_{k=1}^{K}P_{d_k}\right)}}{\sqrt{\nu^*}\left(1+\gamma(\zeta_d^*)\sum_{k=1}^{K}\beta(\zeta_{t_k}^*)\right)}
-T_{t_k}(1+P_{t_k}T_{t_k})\right)^+,\quad k\in\mathcal{K}
\vspace{-10mm}
 \end{equation}
 \setcounter{equation}{\value{tempequationcounter}}
 \vspace*{4pt}
 \end{figure*}
\begin{equation}\label{7}
   R\geq\frac{T_d}{T}\log_2\left(1+\sum_{k=1}^{K}\frac{\frac{P_{d_k}}{1+P_{wd}}\exp(\mathbb{E}\{\log[|\hat{h}_k|^2]\})}{1+\sum_{k=1}^{K}\sigma_{\tilde{h}_k}^2\frac{P_{d_k}}{1+P_{wd}}}\right)=R_{LB}.
\end{equation}
Following \cite{convex}, we know that $\mathbb{E}\{\log[|\hat{h}_k|^2]\}=\log(\sigma_{\hat{h}_k}^2)+\psi(1)=\log(\sigma_{\hat{h}_k}^2)-\kappa$ where $\psi$ is the psi function \cite[Eq. (8.360)]{psi} and $\kappa\approx 0.577$ is the Euler's constant. Hence, we can derive a closed-form solution for $R_{LB}$, which is given by
\begin{equation}\label{8}
R_{LB}=\frac{T-T_t}{T}\log_2\left(1+\rho(\{\zeta_{t_k}\}_{k\in \mathcal{K}},\zeta_d)\exp(-\kappa)\right).
\end{equation}
We observe that the jamming strategy dependent term in both $R_{UB}$ and $R_{LB}$ is the same and given by $\rho(\{\zeta_{t_k}\}_{k\in \mathcal{K}},\zeta_d)$. 

To improve the analytical tractability of the jammer design, we design the jamming energy allocation strategy to minimize both the upper and lower bounds on the maximum achievable sum-rate, instead of the exact sum-rate expression in \eqref{4}. The solution obtained from this design is expected to have near-optimal performance in terms of minimizing the exact sum-rate. In order to minimize $R_{UB}$ and $R_{LB}$, we must minimize $\rho(\{\zeta_{t_k}\}_{k\in \mathcal{K}},\zeta_d)$; since $\log_2(1+ax)$ is a continuous and increasing function of $x$ with $a>0$.
By substituting $\sigma_{\hat{h}_k}^2$ and $\sigma_{\tilde{h}_k}^2$ from \eqref{2} into \eqref{6} and using \eqref{3}, we have
\begin{equation}\label{9}
\rho(\{\zeta_{t_k}\}_{k\in \mathcal{K}},\zeta_d)=\frac{\gamma(\zeta_d)\sum_{k=1}^{K}\alpha(\zeta_{t_k})}{1+\gamma(\zeta_d)\sum_{k=1}^{K}\beta(\zeta_{t_k})}
\end{equation}
where $\gamma(\zeta_d)=\frac{1}{1+\zeta_dP_wT/T_d}$ and
\begin{equation}\label{10}
\alpha(\zeta_{t_k})=\frac{\frac{P_{d_k}P_{t_k}T_{t_k}}{1+\zeta_{t_k}P_wT/T_{t_k}}}{1+\frac{P_{t_k}T_{t_k}}{1+\zeta_{t_k}P_wT/T_{t_k}}},\quad\beta(\zeta_{t_k})=\frac{P_{d_k}}{1+\frac{P_{t_k}T_{t_k}}{1+\zeta_{t_k}P_wT/T_{t_k}}}.
\end{equation}
Thus, we have driven an expression which includes the system parameters of all the users and we use \eqref{9} as the objective function of the optimization problem in the succeeding section.
\vspace{-0.8cm}
\section{Optimal Jamming Energy Allocation}
The optimal jamming energy allocation can be written as the following optimization problem:
\begin{equation}\label{11}
\arg\min_{\{\zeta_{t_k}\}_{k\in\mathcal{K}},\zeta_d} \rho(\{\zeta_{t_k}\}_{k\in \mathcal{K}},\zeta_d) 
\end{equation}
\begin{equation}\label{12}
\text{s.t.}\quad \{\zeta_{t_k}\}_{k\in\mathcal{K}},\zeta_d\geq 0, \sum_{k=1}^{K}\zeta_{t_k}+\zeta_d=1
\end{equation}
where $\rho(\{\zeta_{t_k}\}_{k\in \mathcal{K}},\zeta_d)$ is given by \eqref{9}. The following theorem presents the solution, i.e., $\{\zeta_{t_k}^*\}_{k\in \mathcal{K}}$ and $\zeta_d^*$. 
\addtocounter{equation}{1}
\vspace{-.5cm}
\subsection{Main Result}
\begin{theorem}\label{Th_1}
The optimal jamming energy allocation must satisfy the equations in \eqref{f1}-\eqref{15}, where \eqref{f1} is shown at the bottom of the page\footnote[3]{Here we employ the common notation, $(z)^+=\max(0,z)$.}.
\begin{equation}\label{14}
\zeta_d^*=\frac{1}{P_wT}\left(\frac{\sqrt{P_wTT_d\sum_{k=1}^{K}\alpha(\zeta_{t_k}^*)}}{\sqrt{\nu^*}\left(1+\gamma(\zeta_d^*)\sum_{k=1}^{K}\beta(\zeta_{t_k}^*)\right)}-T_d\right)^+
\end{equation}
\begin{equation}\label{15}
\sum_{k=1}^{K}\zeta_{t_k}^*+\zeta_d^*=1
\end{equation}
We have a system of $K+2$ non-linear equations and we can directly solve them to find the optimal values of $K+2$ variables. If $P_w$ is sufficiently high, all the calculated values of $\{\zeta_{t_k}^*\}_{k\in \mathcal{K}}$ and $\zeta_d^*$ will be positive. In this case, we can derive closed-form solutions for the optimal jamming energy allocation in terms of the training and data transmission parameters of all users, which is given as
\begin{equation}\label{16}
\zeta_{t_k}^*=\frac{P_wT+T+\delta+T_d\sum_{i\in \mathcal{K}}P_{d_k}-\frac{1+P_{t_k}T_{t_k}}{\sqrt{P_{d_k}P_{t_k}}}(2\eta)}{2P_wT(\frac{\eta}{T_{t_k}\sqrt{P_{d_k}P_{t_k}}})},\hspace{1mm} k\in \mathcal{K}
\end{equation}
\begin{equation}\label{17}
\zeta_d^*=\frac{1}{2}+\frac{T_t+\delta-T_d(1+\sum_{i\in\mathcal{K}}P_{d_k})}{2P_wT}
\end{equation}
where 
\begin{equation}\label{18}
\delta=\sum_{i\in \mathcal{K}}P_{t_i}T_{t_i}^2\hspace{2mm}\textnormal{and}\hspace{2mm}\eta=\sum_{i\in \mathcal{K}}T_{t_i}\sqrt{P_{d_i}P_{t_i}}.
\end{equation}
\end{theorem}
\begin{proof}
It can be shown that the Hessian of the objective function in \eqref{9} is positive, so the problem is convex in $\{\zeta_{t_k}\}_{k\in \mathcal{K}}$ and $\zeta_d$. Introducing Lagrange multipliers $\{\lambda_k^*\}_{k\in \mathcal{K}}$ and $\lambda_d^*$ for the inequality constraints and $\nu^*$ for the equality constraint in \eqref{12}, we obtain the KKT conditions as
\begin{equation}\label{19}
\{\zeta_{t_k}^*\}_{k\in\mathcal{K}},\zeta_d^*\geq 0,\quad\sum_{k=1}^{K}\zeta_{t_k}^*+\zeta_d^*=1,\quad\{\lambda_k^*\}_{k\in \mathcal{K}},\lambda_d^*\geq 0
\end{equation}
\begin{equation}\label{20}
\lambda_k^*\zeta_{t_k}^*=0,\quad k\in\mathcal{K},\quad \lambda_d^*\zeta_d^*=0
\end{equation}
\begin{equation}\label{21}
\begin{split}
&-\frac{P_wT\gamma(\zeta_d^*)P_{d_k}P_{t_k}T_{t_k}^2\left(1+\gamma(\zeta_d^*)\sum_{k=1}^{K}P_{d_k}\right)}{\left(P_w\zeta_{t_k}^*T+T_{t_k}(1+P_{t_k}T_{t_k})\right)^2\left(1+\gamma(\zeta_d^*)\sum_{k=1}^{K}\beta(\zeta_{t_k}^*)\right)^2}\\
&\hspace{52mm}-\lambda_k^*+\nu^*=0,\quad k\in\mathcal{K}
\end{split}
 \end{equation}
\begin{equation}\label{22}
-\frac{P_wTT_d\sum_{k=1}^{K}\alpha(\zeta_{t_k}^*)}{(P_w\zeta_d^*T+T_d)^2\left(1+\gamma(\zeta_d^*)\sum_{k=1}^{K}\beta(\zeta_{t_k}^*)\right)^2}-\lambda_d^*+\nu^*=0.
\end{equation}
Note that $\{\lambda_k^*\}_{k\in \mathcal{K}}$ and $\lambda_d^*$ act as slack variables in \eqref{21} and \eqref{22}; hence, they can be eliminated. If we solve \eqref{19}-\eqref{22}, we get to the equations in \eqref{f1}-\eqref{15}.

If the calculated values of $\zeta_{t_i}^*$ and $\zeta_{t_j}^*$ are positive, we conclude from \eqref{f1} that
\begin{equation}\label{23}
\begin{split}
&\frac{P_{d_i}P_{t_i}T_{t_i}^2}{\left(P_w \zeta_{t_i}^*T+(P_{t_i}T_{t_i}+1)T_{t_i}\right)^2}\\
&\hspace{12mm}=\frac{P_{d_j}P_{t_j}T_{t_j}^2}{\left(P_w \zeta_{t_j}^*T+(P_{t_j}T_{t_j}+1)T_{t_j}\right)^2},\quad i,j\in \mathcal{K}.
\end{split}
\end{equation}
As mentioned previously, for sufficiently high $P_w$, all the calculated values of $\{\zeta_{t_k}^*\}_{k\in \mathcal{K}}$ and $\zeta_d^*$ will be positive. In this case, we can calculate the values of $\{\zeta_{t_i}^*\}_{i\in\mathcal{K},i\neq k}$ in terms of $\zeta_{t_k}^*$ by \eqref{23}. By substituting $\alpha(\zeta_{t_k})$ from \eqref{10} into \eqref{14} and using \eqref{23}, we also calculate $\zeta_d^*$ in terms of $\zeta_{t_k}^*$. Now, $\{\zeta_{t_i}^*\}_{i\in\mathcal{K},i\neq k}$ and $\zeta_d^*$ have been calculated in terms of $\zeta_{t_k}^*$. If we substitute them into the equality constraint in \eqref{19}, we get to the equations in \eqref{16} and \eqref{17}.
\end{proof}
\vspace{-0.5cm}
\subsection{Further Analysis}
In the following, we consider various special cases to further analyze the impacts of different system parameters on the optimal jammer strategy.
\begin{corollary}
If $T_{t_k}=T_{t_j}\triangleq T_{t_c}, P_{t_k}=P_{t_j}\triangleq P_{t_c}$ and $P_{d_k}>P_{d_j}$, then $\zeta_{t_k}^*\geq\zeta_{t_j}^*$. 
\end{corollary}
\begin{proof}From \eqref{23} we have
\begin{equation} \nonumber
\frac{P_{d_k}}{P_{d_j}}=\left(\frac{P_w \zeta_{t_k}^*T+(P_{t_c}T_{t_c}+1)T_{t_c}}{P_w \zeta_{t_j}^*T+(P_{t_c}T_{t_c}+1)T_{t_c}}\right)^2>1\Rightarrow\zeta_{t_k}^*>\zeta_{t_j}^*
\end{equation}
and the equality holds when both of them equal zero.
\end{proof}

Although the training parameters of the two users are the same, the optimal jammer allocates more power to jam the training symbols of the user with more data power, so that it can impair the channel estimate of that user more, which results in poorer data detection and lower data rate.
\begin{corollary}
If $T_{t_k}=T_{t_j}\triangleq T_{t_c}, P_{d_k}=P_{d_j}\triangleq P_{d_c}$ and $P_{t_k}>P_{t_j}$, then $\zeta_{t_k}^*\geq\zeta_{t_j}^*$. 
\begin{proof}
From \eqref{23} we have
\begin{equation} \nonumber
\frac{P_{t_k}}{P_{t_j}}=\left(\frac{P_w\zeta_{t_k}^*T-(T_{t_c}\sqrt{P_{t_k}P_{t_j}}-1)T_{t_c}}{P_w\zeta_{t_j}^*T-(T_{t_c}\sqrt{P_{t_k}P_{t_j}}-1)T_{t_c}}\right)^2>1\Rightarrow \zeta_{t_k}^*>\zeta_{t_j}^*
\end{equation}
and the equality holds when both of them equal zero.
\end{proof}
Hence, when the training duration and data power of two users are the same, the optimal jammer allocates more power to jam the training symbols of the user with larger training power to impair its channel estimate more.
\end{corollary}
\begin{figure}[t]
  \centering
  \vspace{-.6cm}
  \includegraphics[width=2.84in]{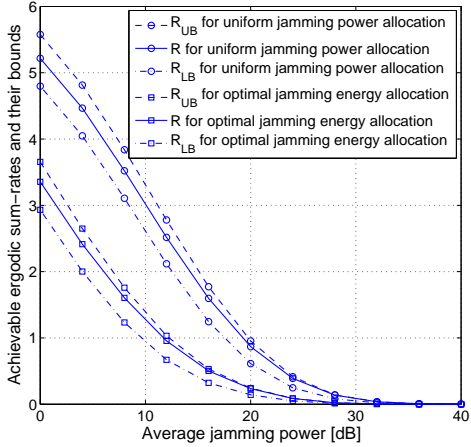}\\
  \vspace{-.5cm}
  \caption{Achievable ergodic sum-rates and their bounds versus the average jamming power $P_w$ for a system of four users with a block length of $T=100$ and training length of $T_t=4$. The legitimate users' energy allocation is set to the optimal case for jamming-free systems. 
  \vspace{-.8cm}
  }\label{fa}
\end{figure}
\begin{corollary}
If $P_{t_k}=P_{t_j}\triangleq P_{t_c}, P_{d_k}=P_{d_j}\triangleq P_{d_c}$ and $T_{t_k}>T_{t_j}$, then $\zeta_{t_k}^*\geq\zeta_{t_j}^*$. 
\begin{proof}
From \eqref{23} we have
\begin{equation} \nonumber
\frac{T_{t_k}}{T_{t_j}}=\frac{P_w\zeta_{t_k}^*T-T_{t_k}T_{t_j}P_{t_c}}{P_w\zeta_{t_j}^*T-T_{t_k}T_{t_j}P_{t_c}}>1\Rightarrow\zeta_{t_k}^*>\zeta_{t_j}^*
\end{equation}
and the equality holds when both of them equal zero.
\end{proof}
Hence, when the training and data power of two users are the same, the optimal jammer allocates more energy to jam the training symbols of the user with larger training duration.
\end{corollary}

\begin{corollary}
If $P_w\rightarrow \infty$, the optimal jamming energy allocation is given by
\begin{equation} \nonumber
\zeta_{t_k}^*=\frac{T_{t_k}\sqrt{P_{t_k}P_{d_k}}}{2\sum_{i=1}^{K}T_{t_i}\sqrt{P_{t_i}P_{d_i}}},\quad\zeta_d^*=\frac{1}{2}.
\end{equation}
\begin{proof}
The proof is straightforward by the equations in \eqref{16} and \eqref{17}.
\end{proof}
\end{corollary} 

\vspace{-0.4cm}
\section{Numerical Results and Discussion}
In this section, we present numerical results on the optimal energy allocation of the smart jammer. 
We have considered four users with average power budgets 5, 10, 15, and 20 dB, respectively. Our model of the legitimate users is a special case of the system model in \cite{Soysal}. Hence, we have used the results of \cite{Soysal} to set the values of  $P_{t_k}$ and $P_{d_k}$ which satisfy the given average power budget and maximize the achievable ergodic sum-rate in jamming-free systems. According to \cite{Soysal}, the optimal training duration for each of the users equals one and hence the total training duration equals four.

Fig.~\ref{fa} shows the performance gain of the optimal jammer over a jammer with uniform power allocation, i.e., $\zeta_{t_k}=T_{t_k}/T$ and $\zeta_d=T_d/T$. We plot $R_{LB}$, $R$ and $R_{UB}$ achieved by using uniform power allocation and also the optimal energy allocation given by Theorem 1. We observe that the data rate degradation caused by the designed jammer is significant.
\begin{figure}[t]
  \centering
  \vspace{-.6cm}
  \includegraphics[width=2.89in]{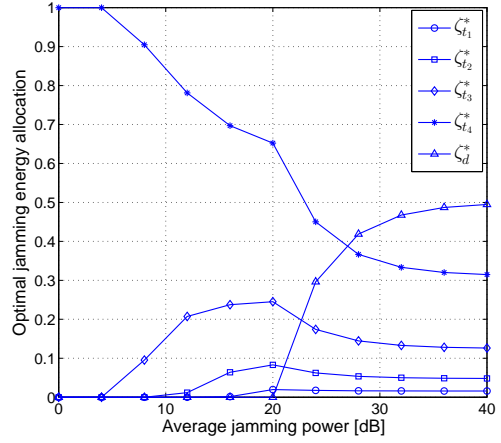}\\
  \vspace{-.5cm}
  \caption{Optimal jamming energy allocation versus the average jamming power $P_w$ for a system of four users with a block length of $T=100$ and training length of $T_t=4$.
  \vspace{-.8cm}
  }\label{fb}
\end{figure}

Fig.~\ref{fb} illustrates the optimal jamming energy allocation versus the average jamming power budget $P_w$. For very small values of $P_w$, the optimal jammer allocates all its power to jam the training symbol of user 4, which has the largest power budget. As $P_w$ increases, the jammer also allocates some power to jam the training signals of users 3, 2, and 1, respectively. Hence, as $P_w$ increases, the training symbols of the users with more power budgets get jammed by the jammer sooner and they are allocated more jamming energy, i.e., $\zeta_{t_4}^*>\zeta_{t_3}^*>\zeta_{t_2}^*>\zeta_{t_1}^*$. The intuitive interpretation for this behaviour is that the jammer tends to allocate more energy to jam the training symbols of the users which contribute more to the ergodic sum-rate. We observe that when $P_w$ is less than 20 dB, all power should be spent on jamming the training phase. If $P_w$ exceeds 20 dB, the jammer starts jamming the data transmission phase and as $P_w$ increases more, the optimal values of $\{\zeta_{t_k}\}_{k\in \mathcal{K}}$ and $\zeta_d$ approach the values given in Corollary 4. Interestingly, the jammer allocates its energy equally between the training and data transmission phases for sufficiently high values of $P_w$. In general, determining the point where the jammer starts jamming the data transmission phase requires solving the $K+2$ non-linear equations in \eqref{f1}-\eqref{15}. However, if $\zeta_d$ is the first among the energy allocation ratios which approaches zero while $P_w$ decreases from higher values, which can be inferred from the equations in \eqref{16} and \eqref{17}, the point where $\zeta_d$ starts to increase can be calculated by setting the equation in \eqref{17} equal to zero and solving the resulting equation for $P_w$.
\vspace{-0.5cm}

\ifCLASSOPTIONcaptionsoff
  \newpage
\fi

\bibliographystyle{IEEEtran}
\bibliography{reference}

\end{document}